\documentclass[11pt]{article}

\usepackage{caption}
\usepackage{algorithm}
\usepackage{algorithmic}
\usepackage{dashrule}
\usepackage{mathrsfs}
\usepackage{enumitem}
\usepackage{bbm}
\usepackage{tikz}
\usepackage{float}
\usetikzlibrary{backgrounds}
\usepackage{color}
\usepackage{graphicx}
\usepackage{latexsym}
\usepackage{amsfonts}
\usepackage{pifont,xspace,epsfig, wrapfig}
\usepackage{amsmath, amsthm}
\usepackage{multirow}
\usepackage{dsfont}
\usepackage{array}
\usepackage{adjustbox}
\usepackage{fullpage}
\usepackage{bookmark}

\DeclareSymbolFont{AMSb}{U}{msb}{m}{n}
\DeclareMathSymbol{\N}{\mathbin}{AMSb}{"4E}
\DeclareMathSymbol{\Z}{\mathbin}{AMSb}{"5A}
\DeclareMathSymbol{\R}{\mathbin}{AMSb}{"52}
\DeclareMathSymbol{\Q}{\mathbin}{AMSb}{"51}
\DeclareMathSymbol{\erert}{\mathbin}{AMSb}{"50}
\DeclareMathSymbol{\I}{\mathbin}{AMSb}{"49}
\DeclareMathSymbol{\C}{\mathbin}{AMSb}{"43}

\definecolor{gray}{gray}{0.4}

\newcommand{\remove}[1]{}

\newtheorem{theorem}{Theorem}[section]
\newtheorem{lemma}[theorem]{Lemma}
\newtheorem{definition}[theorem]{Definition}
\newtheorem{remark}[theorem]{Remark}

\newcommand{\1}{\mathbbm{1}}

\newcommand{\AAA}{\mathcal A}
\newcommand{\BBB}{\mathcal B}
\newcommand{\BbB}{\mathfrak{B}}

\newcommand{\DDD}{\mathcal D}
\newcommand{\FFF}{\mathcal F}

\newcommand{\PPP}{\mathcal P}

\newcommand{\SSS}{\mathcal S}

\newcommand{\eps}{\varepsilon}

\newcommand{\polylog}{\mathop{\rm polylog}}

\newcommand{\poly}{\mathop{\rm poly}}
\def\E{\operatorname*{\mathbb{E}}}

\makeatletter
\newcommand{\thickhline}{%
    \noalign {\ifnum 0=`}\fi \hrule height 1pt
    \futurelet \reserved@a \@xhline
}
\newcolumntype{"}{@{\hskip\tabcolsep\vrule width 1pt\hskip\tabcolsep}}
\makeatother

\makeatletter
\newlength{\fboxhsep}
\newlength{\fboxvsep}
\newlength{\fboxtoprule}
\newlength{\fboxbottomrule}
\newlength{\fboxleftrule}
\newlength{\fboxrightrule}
\def\@frameb@xother#1{%
  \@tempdima\fboxtoprule
  \advance\@tempdima\fboxvsep
  \advance\@tempdima\dp\@tempboxa
  \hbox{%
    \lower\@tempdima\hbox{%
      \vbox{%
        \hrule\@height\fboxtoprule
        \hbox{%
          \vrule\@width\fboxleftrule
          #1%
          \vbox{%
            \vskip\fboxvsep
            \box\@tempboxa
            \vskip\fboxvsep}%
          #1%
          \vrule\@width\fboxrightrule}%
        \hrule\@height\fboxbottomrule}%
    }%
  }%
}
\long\def\fboxother#1{%
  \leavevmode
  \setbox\@tempboxa\hbox{%
    \color@begingroup
    \kern\fboxhsep{#1}\kern\fboxhsep
    \color@endgroup}%
  \@frameb@xother\relax}

\makeatother

\title{Relaxed Models for Adversarial Streaming:\\ The Advice Model and the Bounded Interruptions Model\thanks{This project was partially supported by the Israel Science Foundation (grant 1871/19) and by
Len Blavatnik and the Blavatnik Family foundation.}}

\author{
Moshe Shechner\thanks{Tel Aviv University. \texttt{ moshe.shechner@gmail.com}}
\and
Menachem Sadigurschi\thanks{Ben-Gurion University. \texttt{ menisadi@gmail.com}}
\and
Uri Stemmer\thanks{Tel Aviv University and Google Research. \texttt{u@uri.co.il}}
}

\date{January 22, 2023}

\begin{document}

\maketitle

\begin{abstract}
Streaming algorithms are typically analyzed in the {\em oblivious} setting, where we assume that the input stream is fixed in advance. Recently, there is a growing interest in designing {\em adversarially robust}
streaming algorithms that must maintain utility even when the input stream is chosen
{\em adaptively and adversarially} as the execution progresses. While several fascinating results are known for the adversarial setting, in general, it comes at a very high cost in terms of the required space. Motivated by this, in this work we set out to explore intermediate models that allow us to interpolate between the oblivious and the adversarial models. Specifically, we put forward the following two models:
\begin{itemize}
    \item {\em The advice model}, in which the streaming algorithm may occasionally ask for one bit of {\em advice}. 
    \item {\em The bounded interruptions model}, in which we assume that the adversary is only partially adaptive.
\end{itemize}

We present both positive and negative results for each of these two models. In particular, we present generic reductions from each of these models to the oblivious model. This allows us to design robust algorithms with significantly improved space complexity compared to what is known in the plain adversarial model.
\end{abstract}

\section{Introduction}
Streaming algorithms are algorithms for processing data streams
in which the input is presented as a sequence of items. Generally speaking, these algorithms have access
to limited memory, significantly smaller than what is needed to store the entire data stream. This field was formalized by Alon, Matias, and Szegedy~\cite{AlonMS99}, and has generated a large body of work that intersects many other fields in computer science. 

In this work, we focus on streaming algorithms that aim to track a certain function of the input stream, and to continuously report estimates of this function. Formally,

\begin{definition}[Informal version of Definition~\ref{def:obliviousFull}]\label{def:obliv}
Let $X$ be a finite domain and let $g:X^*\rightarrow \R$ be a function that maps every input $\vec{x}\in X^*$ to a real number $g(\vec{x})\in\R$.

Let $\AAA$ be an algorithm that in every round $i\in[m]$ obtains an element $x_i\in X$ and 
outputs a response $z_i\in\R$. 
Algorithm $\AAA$ is said to be an {\em oblivious streaming algorithm} for $g$ with accuracy $\alpha$, failure probability $\beta$, and stream length $m$, if the following holds for every input sequence $\vec{x}=(x_1,x_2,\dots,x_m)\in X^m$. Consider an execution of $\AAA$ on the input stream $\vec{x}$. Then, 
$$
\Pr\left[\forall i\in[m] \text{ we have } z_i\in (1\pm\alpha)\cdot g(x_1,\dots,x_i) \right]\geq1-\beta,
$$
where the probability is taken over the coins of algorithm $\AAA$.
\end{definition}

Note that in Definition~\ref{def:obliv}, the streaming algorithm is required to succeed (w.h.p.)\ for every {\em fixed} input stream. In particular, it is assumed that the choice for the elements in the stream is {\em independent} from the
internal randomness of the streaming algorithm. This assumption, called the {\em oblivious setting}, is crucial for the correctness of most classical streaming algorithms. In this work, we are interested in the setting where this assumption does not hold, referred to as the {\em adversarial setting}.

\subsection{The (Plain) Adversarial Model}

The adversarial streaming model, in various forms, was considered by~\cite{MironovNS11,GHRSW12,GHSWW12,AhnGM12,AhnGM12b,HardtW13,BenEliezerY19,BenEliezerJWY20,HKMMS-JACM,WZ21,BEO22,ACSS21}. The adversarial setting is modeled by a two-player game between a (randomized) \texttt{StreamingAlgorithm} and an \texttt{Adversary}. At the beginning, we fix a function $g:X^*\rightarrow\R$. Throughout the game, the adversary chooses the updates in the stream, and is allowed to {\em query} the streaming algorithm at $T$ time steps of its choice (referred to as ``query times''). Formally, 

\begin{enumerate}[leftmargin=15px]
    \item For round $i=1,2,\dots,m$
    \begin{enumerate}
	\item The \texttt{Adversary} chooses an update $x_i\in X$ and a query demand $q_i\in\{0,1\}$, under the restriction that $\sum_{j=1}^{i} q_j \leq T$.
	\item The \texttt{StreamingAlgorithm} processes the new update $x_i$. If $q_i=1$ then, the \texttt{Streaming\-Algorithm} outputs a response $z_i$, which is given to the \texttt{Adversary}.
\end{enumerate}
\end{enumerate}

The goal of the \texttt{Adversary} is to make the \texttt{StreamingAlgorithm} output an incorrect response
$z_i$ at some query time $i$ in the stream. 
Let $g$ be a function defining a streaming problem, and suppose that there is an oblivious streaming algorithm $\AAA$ for $g$ that uses space $s$. It is easy to see that $g$ can be solved in the adversarial setting using space $\approx s\cdot T$, by running $T$ copies of $\AAA$ and using each copy for at most one query. The question is if we can do better. Indeed, Hassidim et al.~\cite{HKMMS-JACM} showed the following result.

\begin{theorem}[\cite{HKMMS-JACM}, informal]\label{thm:HKMMS}
If there is an {\em oblivious} streaming algorithm for a function $g$ that uses space $s$, then there is an {\em adversarially robust} streaming algorithm for $g$ supporting $T$ queries using space $\approx\sqrt{T}\cdot s$.
\end{theorem}

Note that when the number of queries $T$ is large, this construction incurs a large space blowup. One way for coping with this is to assume additional restrictions on the function $g$ or on the input stream. Indeed, starting with Ben-Eliezer et al.~\cite{BenEliezerJWY20}, most of the positive results on adversarial streaming assumed that the input stream is restricted to have a small {\em flip-number}, defined as follows. 

\begin{definition}[Flip number \cite{BenEliezerJWY20}]
The $(\alpha, m)$-flip number of 
an input stream $\vec{x}$ w.r.t.\ a function $g$, denoted as $\lambda_{\alpha,m}(\vec{x},g)$, or simply $\lambda$, is the maximal number of times that the value of $g$ changes (increases or decreases) by a factor of $(1+\alpha)$ during the stream $\vec{x}$.
\end{definition}

Starting from \cite{BenEliezerJWY20}, the prior works of \cite{BenEliezerJWY20,HKMMS-JACM,WZ21,ACSS21} presented generic constructions that transform an oblivious streaming algorithm with space $s$ into an adversarially robust streaming algorithm with space $\approx s\cdot\poly(\lambda)$. That is, under the assumption that the flip-number is bounded, these prior works can even support $T=m$ queries. This is useful since the parameter $\lambda$ is known to be small for many interesting streaming problems in the {\em insertion-only} model (where there are no deletions in the stream). However, in general it can be as big as $\Theta(m)$, in which case the transformations of \cite{BenEliezerJWY20,HKMMS-JACM,WZ21,ACSS21} come at a very high cost in terms of space. 

To summarize this discussion, current transformations from the oblivious to the adversarial setting are useful when either the number of queries $T$ is small, or under the assumption that the flip-number is small.

\subsection{Our Results}
One criticism of the adversarial model is that it is (perhaps) too pessimistic. Indeed, there could be many scenarios that do not fall into the oblivious model, but are still quite far from an ``adversarial'' setting. Motivated by this, in this work, we set out to explore intermediate models that allow us to interpolate between the oblivious model and the adversarial model. Specifically, we study two such models, which we call the {\em advice model} and the {\em bounded interruptions model}.

\subsubsection{Adversarial Streaming with Advice (ASA)}\label{sec:ASAintro}

\noindent
We put forward a model where the streaming algorithm may occasionally ask for one bit of {\em advice} throughout the execution. Let $\eta\in\N$ be a parameter controlling the query/advice rate. We consider the following game, referred to as the ASA game, between the \texttt{StreamingAlgorithm} and an \texttt{Adversary}.
\medskip

\noindent\fbox{%
    \parbox{\linewidth}{%
For round $i=1,2,\dots,m:$
    \begin{enumerate}[leftmargin=23px,rightmargin=5px,itemsep=5px]
	\item The \texttt{Adversary} chooses an update $x_i\in X$ and a query demand $q_i\in\{0,1\}$, under the restriction that $\sum_{j=1}^{i} q_j \leq T$.
	\item The \texttt{StreamingAlgorithm} processes the new update $x_i$. 
	
	\item If $q_i=1$ then
	
	\begin{enumerate}[leftmargin=10px,topsep=3px]
	    \item The \texttt{StreamingAlgorithm} outputs a response $z_i$, which is given to the \texttt{Adversary}.
	    
	    \item If $\left(\sum_{j=1}^{i} q_j\right){\rm mod}\,\eta=0$ then the \texttt{StreamingAlgorithm} specifies a predicate $p_i:X^*\rightarrow\{0,1\}$, and  obtains $p_i(x_1,x_2,\dots,x_i)$.
	\end{enumerate}

\end{enumerate}

    }%
}
\; \medskip

That is, in the ASA model the adversary is allowed a total of $T$ queries, and once every $\eta$ queries the streaming algorithm is allowed to obtain one bit of advice, computed as a predicate of the items in the stream so far.

The main motivation to study this model is a theoretical one; it gives us an intuitive way to measure the amount of additional information that the streaming algorithm needs in order to maintain utility in the adversarial setting.
This model could also be interesting from a practical standpoint in the following context. Consider a streaming setting in which the input stream is fed into both a (low space) streaming algorithm $\AAA$ and to a server $\SSS$. The server has large space and can store all the input stream (and, therefore, can in principle solve the streaming problem itself). 
However, suppose that the server has some communication bottleneck and is busy serving many other tasks in parallel. Hence we would like to delegate as much of the communication as possible to the ``cheap'' (low space) streaming algorithm $\AAA$. 
The ASA model allows for such a delegation, in the sense that the streaming algorithm handles most of the queries itself, and only once in every $\eta$ queries it asks for one bit of advice from the server.

We show the following generic result.

\begin{theorem}[informal version of Theorem~\ref{thm:cleanASA}]\label{thm:ASAintro}
If there exists an oblivious {\em linear} streaming algorithm for a function $g:X^*\rightarrow\R$ with space $s$, then for every $\eta\in\N$ there exists an adversarially robust streaming algorithm for $g$ in the ASA model with query/advice rate $\eta$ using space $\approx \eta\cdot s^2$.
\end{theorem}

To obtain this result, we rely on a technique introduced by Hassidim et al.~\cite{HKMMS-JACM} which uses {\em differential privacy} \cite{DMNS06} to protect not the input data, but rather the internal randomness of the streaming algorithm. Intuitively, this allows us to make sure that the ``robustness'' of our algorithm deteriorates {\em slower} than the advice rate, which allows us to obtain an advice-robustness tradeoff.

Note that the space complexity of the algorithm from Theorem~\ref{thm:ASAintro} does not depend polynomially on the number of queries $T$. For example, the following is a direct application of this theorem in the context of $F_2$ estimation (i.e., estimating the second moment of the frequency vector of the input stream).

\begin{theorem}[$F_2$ estimation in the ASA model, informal]
Let $\eta\in\N$.
There exists an adversarially robust $F_2$ estimation
algorithm in the ASA model with query/advice rate $\eta$ that guarantees $\alpha$-accuracy (w.h.p.) using space $\tilde{O}\left(\eta/\alpha^4\right)$.
\end{theorem}

\begin{remark}
We stress that there is a formal sense in which the ASA model is ``between'' the oblivious and the (plain) adversarial models.
Clearly, the ASA model is easier than the plain adversarial model, as we can simply ignore the advice bits. On the other hand, a simple argument shows that the ASA model (with any $\eta>1$) is qualitatively harder than the oblivious setting. To see this, let $\AAA$ be an algorithm in the ASA model for a function $g$ with query/advice rate $\eta>1$. Then $\AAA$ can be transformed into the following oblivious algorithm $\AAA_{\rm oblivious}$ for $g$ (that returns an estimate in every time step without getting any advice):
\begin{enumerate}[leftmargin=15px]
    \item Instantiate $\AAA$. 
    \item In every time $i\in[m]$:
    
    \begin{enumerate}
        \item Obtain an update $x_i\in X$.
        \item Duplicate $\AAA$ (with its internal state) into a shadow copy $\AAA^{\rm shadow}$.
        \item Feed the update $(x_i,0)$ to $\AAA$ and the update $(x_i,1)$ to $\AAA^{\rm shadow}$, and obtain an answer $z_i$ from the shadow copy. Note that we only query the shadow copy.
        \item Output $z_i$ and erase the shadow copy from memory.
    \end{enumerate}
\end{enumerate}
As we ``rewind'' $\AAA$ after every query, it is never expected to issue an advice-request and so $\AAA_{\rm oblivious}$ never issue advice-request as well. That is, $\AAA_{\rm oblivious}$ is an oblivious model algorithm.
Furthermore, a simple argument shows that this algorithm maintains utility in the oblivious setting.\footnote{To see this, fix an input stream $\vec{x}=(x_1,x_2,\dots,x_m)$, and fix $j\in[m]$. Note that the distribution of the output given by $\AAA_{\rm oblivious}$ in time $j$ when running on $\vec{x}$ is identical to the outcome distribution of $\AAA$ when running on the stream $((x_1,0),\dots,(x_{j-1},0),(x_j,1))$, which must be accurate w.h.p.\ by the utility guarantees of $\AAA$ (since there is only 1 query in this alternative stream, then $\AAA$ gets no advice when running on it). The claim now follows by a union bound over the query times.}
\end{remark}

\begin{remark}
Our construction has the benefit that the predicates specified throughout the interaction are ``simple'' in the sense that every single one of them can be computed in a streaming fashion. That is, given the predicate $p_i$, the bit $p_i(x_1,x_2,\dots,x_i)$ can be computed using small space with one pass over $x_1,x_2,\dots,x_i$. 
\end{remark}

\paragraph{A negative result for the ASA model.}
Theorem~\ref{thm:ASAintro} shows a strong positive result in the ASA model, for streaming problems that are defined by real valued functions. We compliment this result by presenting a negative result for a simple streaming problem which is {\em not} defined by a real valued function. Specifically, we consider (a variant of) the well-studied $\ell_0$-sampling problem, where the streaming algorithm must return a {\em uniformly random} element from the set of non-deleted elements. It is known that the $\ell_0$-sampling problem is easy in the oblivious setting (see e.g.\ \cite{JowhariST11}) and hard in the plain adversarial setting (see e.g.\ \cite{AhnGM12}). 
Using a simple counting argument, we show that the $\ell_0$-sampling problem remains hard also in the ASA model {\em even if the query/advice rate is 1}, i.e., even if the streaming algorithm gets an advice bit for {\em every} query.

\subsubsection{Adversarial Streaming with Bounded Interruptions (ASBI)}\label{sec:ASBIintro}

\noindent
Recall that in the plain adversarial model, the adversary is fully adaptive in the sense the $i$th update may be chosen based on all of the information available to the adversary up until this point in time. We consider a refinement of this setting in which the adversary is only {\em partially} adaptive. The game begins with the adversary specifying a complete input stream. Throughout the execution, the adversary (who sees all the outputs given by the streaming algorithm) can adaptively decide to {\em interrupt} and to replace the suffix of the stream (which has not yet been processed by the streaming algorithm). For simplicity, here  we assume that the streaming algorithm is queried on every time step (i.e., $T=m$).

Formally, let $R\in\N$ be a parameter bounding the number interruptions. We consider the following game, referred to as the ASBI game, between the \texttt{StreamingAlgorithm} and an \texttt{Adversary}.
\medskip

\noindent\fbox{%
    \parbox{\linewidth}{%

\begin{enumerate}[leftmargin=15px,rightmargin=5px,itemsep=5px]
    \item The \texttt{Adversary} chooses a stream $\vec{x}=(x_1,x_2,\dots,x_m)\in X^m$.
    \item For round $i=1,2,\dots,m$\\\vspace{-6px}
    
    \begin{enumerate}[leftmargin=10px]

	\item The \texttt{StreamingAlgorithm} obtains the update $x_i$ and outputs a response $z_i$.
	
	\item The \texttt{Adversary} obtains $z_i$, and outputs an interruption demand $d_i\in\{0,1\}$, under the restriction that $\sum_{j=1}^{i} d_j \leq R$.
	
	\item If $d_i=1$ then the adversary outputs a new stream suffix $(x'_{i+1},\dots,x'_m)$ and we override $(x_{i+1},\dots,x_m)\leftarrow(x'_{i+1},\dots,x'_m)$.
\end{enumerate}
\end{enumerate}

    }%
}
\; \medskip

That is, the adversary sees all of the outputs given by the streaming algorithm, and adaptively decides on $R$ places in which it arbitrarily modifies the rest of the stream. Importantly, the streaming algorithm ``does not know'' when interruptions occur. 
This model gives us a very intuitive interpolation between the oblivious setting (in which $R=0$) and the full adversarial setting (obtained by setting $R=m$, or more subtly by setting $R=T$ when there are at most $T\leq m$ queries). We show the following generic result. 

\begin{theorem}[informal version of Theorem~\ref{thm:interruptions}]\label{thm:ASBIintro}
If there exists an oblivious streaming algorithm for a function $g:X^*\rightarrow\R$ using space $s$ then for every $R\in\N$ there exists an adversarially robust streaming algorithm for $g$ in the ASBI model that resists $R$ interruptions using space $\approx R\cdot s$.
\end{theorem}

To obtain this result, we rely on the {\em sketch switching} technique introduced by Ben-Eliezer et al.~\cite{BenEliezerJWY20}. Intuitively, we maintain $2R$ copies of an oblivious streaming algorithm $\AAA$, where in every given moment exactly two of these copies are designated as ``active''. As long as the two active copies produce (roughly) the same estimates, they remain as the ``active'' copies, and we use their estimates as our response. Once they disagree, we discard them both (never to be used again) and designate two (fresh) copies as ``active''.
We show that this construction can be formalized to obtain Theorem~\ref{thm:ASBIintro}.

Note that the space complexity of the algorithm form Theorem~\ref{thm:ASBIintro} does not depend polynomially on the number of time steps $m$.
For example, the following is a direct application of Theorem~\ref{thm:ASBIintro} for $F_2$ estimation.

\begin{theorem}[$F_2$ estimation in the ASBI model, informal]
Let $R\in\N$.
There exists an adversarially robust $F_2$ estimation
algorithm in the ASBI model that guarantees $\alpha$-accuracy (w.h.p.) while resisting $R$ interruptions using space  $\tilde{O}\left(R/\alpha^2\right)$.
\end{theorem}

\paragraph{A negative result for the ASBI model.}
Note that the space blowup of our construction from Theorem~\ref{thm:ASBIintro} grows linearly with the number of interruptions $R$. Recall that in the full adversarial model (where $R=T$ for $T$ queries) it is known that a space blowup of $\sqrt{T}$ suffices (see Theorem~\ref{thm:HKMMS}). Thus, one might guess that the correct dependence in $R$ in the ASBI model should be $\sqrt{R}$. However, we show that this is generally not the case. Specifically, we show that there exists a streaming problem that can easily be solved in the oblivious setting with small space, but necessitates space linear in $R$ in the ASBI model, provided that the number of queries is large enough (polynomial in $R$).

\subsection{Additional Related Works}
The adversarial streaming model (in a setting similar to ours) dates back to at least \cite{AhnGM12}, who studied it implicitly and showed an impossibility result for robust $\ell_0$ sampling in sublinear memory. The adversarial streaming model was then formalized explicitly by \cite{HardtW13}, who showed strong impossibility results for linear sketches. A recent line of work, 
starting with \cite{BenEliezerJWY20} and continuing with \cite{HKMMS-JACM,WZ21,ACSS21,BEO22} showed {\em positive} results (i.e., {\em robust algorithms}) for many problems of interest, under the assumption that the {\em flip-number} of the stream is bounded.
On the negative side, \cite{BenEliezerJWY20} also presented an attack with $O(n)$ number of adaptive rounds on a variant of the AMS sketch, where $n$ is the size of the domain.
Later, \cite{KaplanMNS21} constructed a streaming {\em problem} for which every adversarially-robust streaming algorithm must use polynomial space, thus showing a separation between the oblivious model and the (plain) adversarial model. 
More recently, \cite{DBLP:conf/icml/Cohen0NSSS22} presented an attack on a concrete algorithm, namely \texttt{CountSketch}, that has length that is {\em linear} in the space of the algorithm and is using only two rounds of adaptivity.

\section{Preliminaries}

In this work we consider streaming problems which are defined by a real valued function (in which case the goal is to approximate the value of this function) as well as streaming problems that define set of valid solutions and the goal is to return one of the valid solutions. The following definition unifies these two objectives for the oblivious setting.

\begin{definition}[Oblivious streaming]\label{def:obliviousFull}
Let $X$ be a finite domain and let $g:X^*\rightarrow 2^W$ be a function that maps every input $\vec{x}\in X^*$ to a subset $g(\vec{x})\subseteq W$ of valid solutions (from some range $W$).

Let $\AAA$ be an algorithm that, for $m$ rounds, obtains an element $x_i\in X$ and outputs a response $z_i\in W$. 
Algorithm $\AAA$ is said to be an {\em oblivious streaming algorithm} for $g$ with failure probability $\beta$, and stream length $m$, if the following holds for every input sequence $\vec{x}=(x_1,x_2,\dots,x_m)\in X^m$. Consider an execution of $\AAA$ on the input stream $\vec{x}$. Then, 
$$
\Pr\left[\forall i\in[m] \text{ we have } z_i\in g(x_1,\dots,x_i) \right]\geq1-\beta,
$$
where the probability is taken over the coins of algorithm $\AAA$.
\end{definition}

For example, in the problem of estimating the number of distinct elements in the stream, the function $g$ in the above definition returns the interval $g(x_1,\dots,x_i)=(1\pm\alpha)\cdot|\{x_1,\dots,x_i\}|$, where $\alpha$ is the desired approximation parameter.

\subsection{Preliminaries from Differential Privacy}
Differential privacy \cite{DMNS06} is a mathematical definition for privacy that aims to enable statistical analyses of databases while providing strong guarantees that individual-level information does not leak. Consider an algorithm $\mathcal{A}$ that operates on a database in which every row represents the data of one individual. Algorithm $\mathcal{A}$ is said to be {\em differentially private} if its outcome distribution is insensitive to arbitrary changes in the data of any single individual. Intuitively, this means that algorithm $\mathcal{A}$ leaks very little information about the data of any single individual, because its outcome would have been distributed roughly the same even without the data of that individual. Formally,
\begin{definition}[\cite{DMNS06}]Let $\mathcal{A}$ be a randomized algorithm that operates on databases. Algorithm $\mathcal{A}$ is $(\eps,\delta)$-{\em differentially private} if for any two databases $S,S^{\prime}$ that differ on one row, and any event $T$, we have
$$
\Pr \left[ \mathcal{A}(S)\in T \right] \leq e^{\eps} \cdot \Pr \left[ \mathcal{A}(S^{\prime})\in T \right] + \delta.
$$
\end{definition}

\subsubsection{Privately Approximating the Median of the Data}

\noindent
Given a database $S\in X^*$, consider the task of {\em privately} identifying an {\em approximate median} of $S$. Specifically, for an error parameter $\Gamma$, we want to identify an element $x\in X$ such that there are at least $|S|/2-\Gamma$ elements in $S$ that are larger or equal to $x$, and there are at least $|S|/2-\Gamma$ elements in $S$ that are smaller or equal to $x$. The goal is to keep $\Gamma$ as small as possible, as a function of the privacy parameters $\eps,\delta$, the database size $|S|$, and the domain size $|X|$.

There are several advanced constructions in the literature with error that grows very slowly as a function of the domain size (only polynomially with $\log^*|X|$)~\cite{BNS13b,BNSV15,BunDRS18,KaplanLMNS19,CohenLNSS23}. 
In our application, however, simpler constructions suffice (where the error grows logarithmically with the domain size).
The following theorem can be derived as  an immediate application of the exponential mechanism \cite{MT07}.

\begin{theorem}\label{thm:Pmed}
There exists an $(\eps,0)$-differentially private algorithm that given a database $S\in X^*$ outputs an element $x\in X$ such that with probability at least $1-\beta$ there are at least $|S|/2-\Gamma$ elements in $S$ that are bigger or equal to $x$, and there are at least $|S|/2-\Gamma$ elements in $S$ that are smaller or equal to $x$, where $\Gamma=O\left(\frac{1}{\eps}\log\left(\frac{|X|}{\beta}\right)\right)$.
\end{theorem}

\subsubsection{Composition of Differential Privacy}

\noindent
The following theorem allows arguing about the privacy guarantees of an algorithm that accesses its input database using several differentially private mechanisms.

\begin{theorem}[\cite{DRV10}]\label{thm:composition2}
Let $0<\eps,\delta'\leq1$, and let $\delta\in[0,1]$. A mechanism that permits $k$ adaptive interactions with mechanisms that preserve $(\eps,\delta)$-differential privacy (and does not access the database otherwise) ensures $(\eps', k\delta+\delta')$-differential privacy, for $\eps'=\sqrt{2k\ln(1/\delta')}\cdot\eps+2k\eps^2$.
\end{theorem}

\subsubsection{The Generalization Properties of Differential Privacy}

\noindent
Dwork et al.~\cite{DFHPRR14} and Bassily et al.~\cite{BassilyNSSSU16} showed that if a predicate $h$ is the result of a differentially private computation on a random sample, then the empirical average of $h$ and its expectation over the underlying distribution are guaranteed to be close. Formally,

\begin{theorem}[\cite{DFHPRR14,BassilyNSSSU16,StemmerPhD}]\label{thm:DPgeneralization}
Let $\eps \in (0,1/3)$, $\delta \in (0,\eps/4)$ and $n\geq\frac{1}{\eps^2}\log(\frac{2\eps }{\delta})$.
Let $\AAA:X^n\rightarrow2^X$ be an $(\eps,\delta)$-differentially private algorithm that operates on a database of size $n$ and outputs a predicate $h:X\rightarrow\{0,1\}$. Let $\DDD$ be a distribution over $X$, let $S$ be a database containing $n$ i.i.d.\ elements from $\DDD$, and let $h\leftarrow\AAA(S)$. Then 
$$
\Pr_{\substack{S\sim\DDD^n \\ h\leftarrow \AAA(S)}}\left[ 
\left| \frac{1}{|S|}\sum_{x\in S}h(x) - \E_{x\sim\DDD}[h(x)] \right|
\geq 10\eps \right] < \frac{\delta}{\eps}.
$$
\end{theorem}

\section{Adversarial Streaming with Advice (ASA)}

In this section we present our results for the ASA model, defined in Section~\ref{sec:ASAintro}. We begin with our generic transformation.

\subsection{A Generic Construction for the ASA Model}

\begin{algorithm*}[t!]
\caption{\bf \texttt{RobustAdvice}($\beta,m,\eta$)}
\begin{flushleft}
{\bf Input:} Parameters: $\beta$ is the failure probability, $m$ is the length of input stream and $\eta$ is the advice query cycle.

{\bf Algorithm used:} An oblivious linear streaming algorithm $\AAA$ with space $s$ for $\alpha$ accuracy.

{\bf Constants calculation:}
\begin{enumerate}
[leftmargin=25pt,rightmargin=10pt,itemsep=1pt,topsep=0pt]
	
	\item $k=\Omega(\eta s \log(m/\beta)\log^2(m/(\beta\alpha)))$ is the number of instances of each of the sets `active', `next' and `shadow'.
	
	\item $\eps_0=\frac{\eps}{\sqrt{8\eta k s\ln(1/\delta)}}$ is the privacy parameter of \texttt{PrivateMed} executions, where $\eps = 1/100$, $\delta = O(\beta/m)$.
	
\end{enumerate} 

\end{flushleft}

\hdashrule[0.5ex][x]{\linewidth}{0.5pt}{1.5mm}

\begin{enumerate}[leftmargin=15pt,rightmargin=10pt,itemsep=1pt,topsep=3pt]
\item 
Initialize $k$ independent instances $\AAA^{\rm active}_1,\dots,\AAA^{\rm active}_k$ of algorithm $\AAA$.

\item\label{step:while} REPEAT (outer loop)

\begin{enumerate}
	\item Initialize $k$ independent instances $\AAA^{\rm next}_1,\dots,\AAA^{\rm next}_k$ of algorithm $\AAA$. \label{algstep:outerLoopInit}
	
	\item Let $\AAA^{\rm shadow}_1,\dots,\AAA^{\rm shadow}_k$ be duplicated copies of $\AAA^{\rm next}_1,\dots,\AAA^{\rm next}_k$, where each $\AAA^{\rm shadow}_j$ is initiated with the same randomness as $\AAA^{\rm next}_j$.
	\label{algstep:outerLoopDuplicate}
	
	\item Denote the current time step as $t$. (That is, so far we have seen $t$ updates in the stream.)
	
	\item REPEAT (inner loop)

	\begin{enumerate}
		\item Receive next update $x_i$ and a query demand $q_i\in\{0,1\}$.
		
		\item Insert update $x_i$ into each of $\AAA^{\rm active}_1,\AAA^{\rm next}_1,\dots,\AAA^{\rm active}_k,\AAA^{\rm next}_k$. 
		
		\item If $q_i=1$ then:
		
		\begin{itemize}
		    \item Query $\AAA^{\rm active}_1,\AAA^{\rm active}_2,\dots,\AAA^{\rm active}_k$ and obtain answers $y_{i,1},y_{i,2},\dots,y_{i,k}$ 
		    
		    \item Output $z_i\leftarrow\texttt{PrivateMed}(y_{i,1},y_{i,2},\dots,y_{i,k})$ with privacy parameter $\eps_0$.
		    
		    \item If $\left(\sum_{j=1}^{i} q_j\right){\rm mod}\,\eta=0$ then define the predicate $p_i$ that given a (prefix of a) stream $\vec{u}$ returns the next bit in the inner state of $(\AAA^{\rm shadow}_1,\AAA^{\rm shadow}_2,\dots,\AAA^{\rm shadow}_k)$ after processing the first $t$ updates in $\vec{u}$. Update the corresponding bit in the state of the corresponding $\AAA^{\rm shadow}_j$.
		    
		    \item If $\left(\sum_{j=1}^{i} q_j\right){\rm mod}\,\eta k s=0$ then EXIT inner loop. Otherwise, CONTINUE inner loop.
		    
		\end{itemize}

	\end{enumerate}
	
	\item For $j\in[k]$ let $\AAA^{\rm active}_j\leftarrow\AAA^{\rm next}_j+\AAA^{\rm shadow}_j$. \label{alg:stateRecovery}
	
\end{enumerate}

\end{enumerate}
\end{algorithm*}

Our generic construction for the ASA model transforms an oblivious and {\em linear} streaming algorithm $\AAA$ into a robust streaming algorithm in the ASA model. The linearity property that we need is the following. Suppose that three copies of $\AAA$, call them $\AAA_1, \AAA_2, \AAA_3$, are instantiated with the {\em same} internal randomness $r$, and suppose that $\AAA_1$ processes a stream $\vec{x}_1$ and that $\AAA_2$ processes a stream $\vec{x}_2$ and that $\AAA_3$ processes the stream $\vec{x}_1 \circ \vec{x}_2$. Then there is an operation (denote it as ``+'') that allows us to obtain an internal state $(\AAA_1 + \AAA_2)$ that is identical to the internal state of $\AAA_3$. Many classical streaming algorithms have this property (for example, the classical AMS sketch for $F_2$ has this property~\cite{AlonMS99}). 
Formally,

\begin{definition}[Linear state algorithm] Let $\mathcal{A}$ be an algorithm with the three segments of memory state. The first segment is randomized in the beginning of the algorithm and remain fixed throughout its execution and denoted as $\mathcal{S}_R$. The second segment is an encoding of vector in $\R^{d}$, and denoted $\mathcal{S}_v$. The third segment is the rest of its memory space and denoted as $\mathcal{S}_c$ and is used for other computations. Then, $\mathcal{A}$ is {\em linear state} w.r.t.~its input stream if for any two streams $\vec{u}_1=((x_1,q_1),\dots,(x_l,q_l))\in (X\times\{0,1\})^l$, $\vec{u}_2=((x_1,q_1),\dots,(x_p,q_p))\in (X\times\{0,1\})^p$ with length of $l,p\in \N$ and three different executions of $\mathcal{A}$ with the same randomized state ($\mathcal{S}_R$) the following holds:
$$\mathcal{S}_v(\vec{u}_1 \circ \vec{u}_2) = \mathcal{S}_v(\vec{u}_1)+\mathcal{S}_v(\vec{u}_2)$$
Where $\mathcal{S}_v(\vec{u})$ is the encoded vector $v\in \R^d$ resulting from the input stream $\vec{u}$ encoded in the corresponding memory state.
\end{definition}

Consider algorithm \texttt{RobustAdvice}. 
In the beginning of each iteration of the outer loop, algorithm \texttt{RobustAdvice} samples $k$ fresh random strings, with which it instantiates the corresponding {\em next} and {\em shadow} copies of algorithm $\AAA$ for that outer loop iteration. Denote by $\tau$ the number of such outer loops ($\tau \leq m$). Instead of sampling these random strings, let us imagine that algorithm \texttt{RobustAdvice} gets these strings as inputs in the following format:
$$\vec{r}=\left(\vec{r^1},\vec{r^2},\dots,\vec{r^{\tau}}\right)=\left( ( r^1_1,\dots,r^1_k) , ( r^2_1,\dots,r^2_k) ,\dots, ( r^{\tau}_1,\dots,r^{\tau}_k) \right),$$
where $\vec{r^j}=( r^j_1,\dots,r^j_k)$ for some $j\in [\tau]$ is used as the random strings for the iteration of the outer loop for the $k$ instances initialized on Step $\ref{algstep:outerLoopInit}$ and their $k$ duplicates at step \ref{algstep:outerLoopDuplicate}. 
Denote by $(t_1<\dots<t_{\tau})$ the times that each outer loop began (Step $\ref{algstep:outerLoopInit}$) and by $\vec{x}_t$ the input stream till time $t$.
Then the following lemma states that at the beginning of the outer loop iteration (say, time $t_j$ for some $j\in [\tau]$) the instances of $\mathcal{A}^{\rm active}_i$ for $i\in [k]$ consist of the state that corresponds to the input stream in time $t_j$, $\vec{x}_{t_j}$, for a random string $r^{t_{j-1}}_i$. That is, Step~\ref{alg:stateRecovery} has successfully recovered the state of instance $i$ of algorithm $\mathcal{A}^{\rm active}$ w.r.t. $r_i^{t_{j-1}}$ for input $\vec{x}_{t_j}$.

\begin{lemma}[State recovery]\label{lem:stateRecovery}
Denote by $\tau$ the number of outer loops, by $(t_1<\dots<t_{\tau})\in [m]^{\tau}$ the times of the beginning of each outer loop (Step~\ref{algstep:outerLoopInit}) and by $\vec{x}_t$ the input stream till time $t$. Then, for all $j\in \{1,\dots,\tau-1\}$, for all $i\in [k]$ algorithm instance $\mathcal{A}^{\rm active}_i$ consist of state segment $\mathcal{S}_v(\vec{x}_{t_j})$ corresponding to state segment $\mathcal{S}_{R}(r_i^{j-1})$.
\end{lemma}

\begin{proof} Fix $j\in \{1,\dots,\tau-1\}$. In time $t_{j-1}$ the strings $r_i^{t_{j-1}}$ was used for the randomization of both $\mathcal{A}^{\rm shadow}_i, \mathcal{A}^{\rm next}_i$ for $i\in [k]$. We focus on the execution of the inner loop, that is time steps $t\in [t_{j-1}, t_{j})$. Throughout the execution of time steps $t\in [t_{j-1}, t_{j})$, the state $\mathcal{S}_v$ of algorithms $\mathcal{A}^{\rm next}$ are updated by the stream update of these time steps. 
In addition, during time steps $t\in [t_{j-1}, t_{j})$, algorithm \texttt{RobustAdvice} issues advice queries that are corresponding to the input stream $\vec{x}_{t_{j-1}}$ with randomization strings $\vec{r^{j-1}}$, where each query recovers additional bit from the $k$ $\mathcal{S}_v$-states of the $\mathcal{A}^{\rm shadow}$ instances. 
Now, since the inner loop is of size at least $\eta\cdot k \cdot s$ where $s$ is the number of bits of the state $\mathcal{S}_v$, then during time steps $t\in [t_{j-1}, t_{j})$ all of the $s$ bits of state $\mathcal{S}_v$ of all of the $k$ instances of $\mathcal{A}^{\rm shadow}$ are recovered via these advice queries. 
That is, each of the instances $\mathcal{A}^{\rm shadow}_i$ consist of $\mathcal{S}_v$ corresponding to $\vec{x}_{t_{j-1}}$ with randomization string $r^{j-1}_i$, for $i\in [k]$.
Recall that for $i\in [k]$ each of the pairs $(\mathcal{A}^{\rm shadow}_i,\mathcal{A}^{\rm next}_i)$ have state $\mathcal{S}_v$ corresponding to the same randomization $r^{j-1}_i$.
And so by the linearity of algorithm $\mathcal{A}$ we have that summing the state $\mathcal{S}_v$ for each of the pairs $(\mathcal{A}^{\rm shadow}_i,\mathcal{A}^{\rm next}_i)$ in Step~\ref{alg:stateRecovery} resulting in state $\mathcal{S}_v$ of the instances $\mathcal{A}^{\rm active}$ corresponding to the input stream $\vec{x}_{t_j}$. 
That is, we have that all of the $\mathcal{A}^{\rm active}$ instances have recovered their $\mathcal{S}_v$ for time $t_j$.
\end{proof}

\begin{lemma}\label{lem:privacy} Let $\tau$ be the number of outer loops of algorithm \texttt{RobustAdvice} and denote by $\vec{r^{j}}$ for $j\in [\tau]$ the random strings vector of the outer loop number $j$.
Then for every $j\in[\tau]$, algorithm \texttt{RobustAdvice} is $(\eps,\delta)$-differential privacy w.r.t.\ $\vec{r^j}$.
\end{lemma}

\begin{proof} Denote by $\tau$ the number of outer loops and denote by $t_1, ..., t_{\tau}$ the time steps that the algorithm executes \ref{algstep:outerLoopInit}. In each such time $t_j\in \{t_1, ..., t_{\tau}\}$ algorithm \texttt{RobustAdvice} uses a new database: throughout time steps $t\in [t_j, t_{j+1})$ the algorithm is using $\vec{r^{j}}$ exclusively. Now, fix $j\in [\tau]$. For time steps $t\in [t_j, t_{j+1})$ algorithm \texttt{RobustAdvice} is executing $\eta k s$ times the \texttt{PrivateMed} mechanism. By composition (Theorem \ref{thm:composition2}) and \texttt{PrivateMed} guarantee (Theorem \ref{thm:Pmed}), selecting $\eps_{0}=\eps/\sqrt{8\eta k s \log (1/\delta)}$ assures that each inner loop $j\in [\tau]$ is $(\eps,\delta)$-DP w.r.t.~its exclusive database $\vec{r^{j}}$.
\end{proof}

\begin{theorem}[Algorithm \texttt{RobustAdvice} is robust]\label{thm:ASA-bigTheorem} Denote by $\vec{x}_{t}$ the input stream till time $t$. Provided that $\mathcal{A}$ is an oblivious linear algorithm for a real valued function $g$, s.t.~w.p. at least $9/10$ it is accurate for all $t\in [m]$ in the oblivious setting (i.e.~it gives estimations $\hat{g}(\vec{x}_t)\in (1\pm \alpha)\cdot g(\vec{x}_t)$), then w.p.~at least $1-\beta$ for all $t\in [m]$ we have:
$$z_t \in (1\pm \alpha)\cdot g(\vec{x}_t) \text{.}$$
\end{theorem}

\begin{proof} Let $t_1< \dots < t_{\tau}$ be the times when the outer loop begins (Step~\ref{algstep:outerLoopInit}). Fix $j\in [\tau]$, and focus on time segment $t\in [t_{j-1},t_{j})$ during which the $j$'th outer loop iteration is executed. 
Let $\vec{x}_t=(x_1,\dots,x_t)$ denote the first $t$ updates in the stream.
By Lemma \ref{lem:stateRecovery} we have that on time $t_{j-1}$ all of the $k$ instances of $\mathcal{A}^{\rm active}$ are updated w.r.t.~the input stream $\vec{x}_{t_{j-1}}$ each corresponding to the random string $r^{j-1}_i$ for $i\in [k]$.
We now argue that these instances remain robust throughout the time segment $t\in [t_{j-1},t_{j})$ in which they receive the input stream updates of this time segment and the output of \texttt{RobustAdvice} is a function of these instances estimations.
Let $\AAA(r,\vec{x}_t)$ denote the output of the oblivious algorithm $\AAA$ when it is instantiated with the random string $r$ and queried after seeing the stream $\vec{x}_t$. Consider the following function:
$$
f_{\vec{x}_t}(r)=\1\left\{  \AAA(r,\vec{x}_t)\in\left(1\pm\alpha\right) \cdot g(\vec{x}_t) \right\}.
$$
By Lemma~\ref{lem:privacy} algorithm \texttt{RobustAdvice} is $\left(\eps{=}\frac{1}{100},\delta=\frac{\eps \beta}{2m}\right)$ - differentially private w.r.t.\ the collection of strings $\vec{r^{j-1}}$. Furthermore, the updates in the stream $\vec{x}_t$ are chosen (by the adversary) by post-processing the estimates returned by \texttt{RobustAdvice}, and the function $f_{\vec{x}_t}(\cdot)$ is defined by $\vec{x}_t$. As differential privacy is closed under post-processing, we can view the {\em function} $f_{\vec{x}_t}(\cdot)$ as the outcome of a differentially private computation on the collection of strings $\vec{r^{j-1}}$.
Therefore, by the generalization properties of differential privacy (Theorem~\ref{thm:DPgeneralization}), assuming that $k\geq\frac{1}{\eps^2}\log(\frac{2\eps }{\delta})$, with probability at least $(1-\frac{\delta}{\eps})$, for every $t\in [t_{j-1}, t_j)$ it holds that
$$
\left| \E_{r}[f_{\vec{x}_t}(r)] - \frac{1}{k}\sum_{i=1}^k f_{\vec{x}_t}(r^{j-1}_i) \right|\leq 10\eps=\frac{1}{10}.
$$
This holds for any $j\in [\tau]$, and so holds for all $t\in [m]$ w.p. at least $1-\beta/2$ by selecting $\frac{\delta}{\eps} = \frac{\beta}{2m}$ (for the corresponding random strings and functions $f_{\vec{x}_t}(r)$).
We continue  the analysis assuming that this is the case.
Now observe that $\E_{r}[f_{\vec{x}_t}(r)]\geq 9/10$ by the utility guarantees of $\AAA$ (because when the stream is fixed its answers are accurate to within multiplicative error of $(1\pm\alpha)$ with probability at least $9/10$). Thus, for at least $(\frac{9}{10}-10\eps)k=4k/5$ of the executions of $\AAA$ we have that $f_{\vec{x}_t}(r_i)=1$, which means that $y_{t,i}\in(1\pm\alpha)\cdot g(\vec{x}_t)$. That is, in every time step $t\in[m]$ we have that at least $4k/5$ of the $y_{t,i}$'s satisfy $y_{t,i}\in(1\pm\alpha)\cdot g(\vec{x}_t)$.
Recall that when the algorithm outputs an estimate, it is computed using algorithm \texttt{PrivateMed}, which is executed on the database $(y_{t,1},\dots,y_{t,k})$. 
By Theorem~\ref{thm:Pmed}, assuming that\footnote{We assume that the estimates that $\AAA$ returns are in the range $[-n^c,-1/n^c]\cup\{0\}\cup[1/n^c,n^c]$ (polynomially bounded in $n$) for some constant $c>0$. In addition, before running \texttt{PrivateMed} we may round each $y_{t,i}$ to its nearest power of $(1+\alpha)$, which has only a small effect on the error. There are at most $O(\frac{1}{\alpha}\log n)$ possible powers of $(1\pm\alpha)$ in that range, and hence, \texttt{PrivateMed} guarantees error at most $\Gamma=O(\frac{1}{\eps_0}\log\left(\frac{\lambda}{\alpha\delta}\log n\right))$ with probability at least $1-\delta/\lambda$. See Theorem~\ref{thm:Pmed}.
Recall also our assumption that $\log m=\Theta(\log n)$. }
$$k
=\Omega\left(\frac{1}{\eps_0}\log \left(\frac{m}{\beta}\frac{\log m}{\alpha}\right)\right)
=\Omega\left(\sqrt{\eta k s\cdot\log\left(\frac{m}{\beta}\right)}\cdot\log \left(\frac{m}{\beta\alpha}\right)\right),
$$
then with probability at least $1-\frac{\beta}{2m}$ algorithm \texttt{PrivateMed} returns an approximate median $\tilde{g}$ for the estimates $y_{t,1},\dots,y_{t,k}$, satisfying
$$
\left|\left\{ i : y_{t,i}\geq\tilde{g} \right\}\right|\geq\frac{4k}{10} 
\qquad\text{ and }\qquad 
\left|\left\{ i : y_{t,i}\leq\tilde{g} \right\}\right|\geq\frac{4k}{10}.
$$
Since $4k/5$ of the $y_{i,j}$'s satisfy $y_{t,i}\in(1\pm\alpha)\cdot g(\vec{x}_t)$, 
such an approximate median $z_t$ must also be in the range $(1\pm\alpha)\cdot g(\vec{x}_t)$. 
This holds separately for every estimate computed in time $t\in [m]$ (approximated median $z_t$) with probability at least $1-\frac{\beta}{2m}$, thus holds simultaneously for all the estimates computed throughout the execution with probability at least $1-\beta/2$.
Overall, we have robustness for all $t\in [m]$ w.p.~at least $1-\beta/2$ and \texttt{PrivateMed} for all $t\in [m]$ executed within its error guarantee w.p.~at least $1-\beta/2$, and so we have that w.p.~at least $1-\beta$ for all $t\in [m]$ \texttt{RobustAdvice} output $z_t$ admits: $$z_t \in \left(1\pm\alpha\right)\cdot g(\vec{x}_t).$$
\end{proof}

The following Theorem now follows from Theorem~\ref{thm:ASA-bigTheorem}.

\begin{theorem}\label{thm:cleanASA}
Fix any real valued function $g$ and fix $\alpha,\beta>0$ and $\eta\in\N$. Let $\AAA$ be an oblivious linear streaming algorithm for $g$ that uses space $s$ and guarantees accuracy $\alpha$ with failure probability $1/10$. Then there exists an adversarially robust streaming algorithm for $g$ in the ASA model with query/advice rate $\eta$, accuracy $\alpha$, and failure probability $\beta$ using space $O(\eta s^2 \log(m/\beta)\log^2(m/(\beta\alpha)) )$.
\end{theorem}

\subsection{A Negative Result for the ASA Model}

In this section we show that $\ell_0$-sampling, a classical streaming problem, cannot be solved with sublinear space in the adversarial setting with advice.
Consider a turnstile stream 
$\vec{u} = (u_1, \dots , u_m)$ where each $u_i = (a_i, \Delta_i) \in [n]\times\{\pm1\}$. A $\beta$-error $\ell_0$-sampler returns with probability at least $1-\beta$ a uniformly random element from
$${\rm support}(u_1, \dots , u_m) = \left\{a\in[n] : \sum_{i:a_i=a}\Delta_i \neq 0\right\},$$
provided that this support is not empty.
The next theorem, due to Jowhari et al.~\cite{jowhari2011tight}, shows that $\ell_0$ sampling is easy in the oblivious setting.

\begin{theorem}[\cite{jowhari2011tight}]
There is a streaming algorithm with storage
$O\left(\log^2(n) \log(\frac{1}{\beta})\right)$ bits, that with probability at most $\beta$ reports FAIL, with probability at most $1/n^2$ reports an arbitrary answer, and in all other cases produces a uniform sample from ${\rm support}(\vec{u})$.
\end{theorem}

Nevertheless, as we next show, this is a hard problem in the ASA setting. In fact, our negative result even holds for a simpler variant of the $\ell_0$ sampling problem, in which the algorithm is allowed to return an {\em arbitrary} element, rather than a random element. Formally,

\begin{definition}
Let $X$ be a finite domain and 
let $\AAA$ be an algorithm that operates on a stream of updates $(u_1,\dots,u_m)\in (X\times\{\pm1\})$, given to $\AAA$ one by one. Algorithm $\AAA$ solves the $J_0$ problem with failure probability $\beta$ if, except with probability at most $\beta$, whenever $\AAA$ is queried it outputs an element with non-zero frequency w.r.t.\ the current stream. That is, if $\AAA$ is queried in time $i$ then it should output an element from ${\rm support}(u_1,
\dots,u_i)$.
\end{definition}

\begin{theorem}\label{thm:streamingNegative}
Let $X$ be a finite domain and let $T$ be such that $|X|=\Omega(T)$ (large enough). 
Let $\AAA$ be an algorithm for solving the $J_0$ streaming problem over $X$ in the adversarial setting with advice with $T$ queries and with failure probability at most $3/4$. Then $\AAA$ uses space $\Omega(T)$. 
Furthermore, this holds also when $\eta=1$, that is, even if algorithm $\AAA$ gets an advice after {\em every} query.
\end{theorem}

\begin{proof}
Let $\AAA$ be an algorithm for $J_0$ sampling with $T$ queries over domain $X$ in the ASA setting with failure probability at most $3/4$. Consider the following thought experiment.\\

\noindent\fbox{%
    \parbox{\linewidth}{%
        \;\\
        {\bf Input:} $Y\subseteq X$ of size $|Y|=T$ 
        \begin{enumerate}
            \item For every $x\in Y$, feed algorithm $\AAA$ the update $(x,1)$.
            \item Initiate $\hat{Y}=\emptyset$.
            \item Repeat $T$ times: 
            \begin{enumerate}
                \item Query $\AAA$ and obtain an outcome $x\in X$
                \item If $\AAA$ requests an advice then give it a random bit $b$.
                \item Add $x$ to $\hat{Y}$
                \item Feed the update $(x,-1)$ to $\AAA$
            \end{enumerate}
            \item Output $\hat{Y}$.
        \end{enumerate}
    }%
}
\; \medskip

We say that the thought experiment {\em succeeds} if $\hat{Y}=Y$. By the assumption on $\AAA$, for every input $Y$, our thought experiment succeeds with probability at least $2^{-T}/4$. This is because if all of the bits of advice are correct then $\AAA$ succeeds with probability at least $1/4$, and the advice bits are all correct with probability at least $2^{-T}$. Hence, there must exist a fixture of $\AAA$'s coins and a fixture of an advice string $\vec{b}$ for which our thought experiments succeeds on at least $2^{-T}/4$ fraction of the possible inputs $Y$.\footnote{Otherwise, consider sampling an input $Y$ uniformly. We have that $\frac{2^{-T}}{4}\leq\Pr_{r,\vec{b},Y}[\AAA_{r,\vec{b}}(Y) \text{ succeeds}]=\sum_{r,\vec{b}}\Pr[r,\vec{b}]\cdot\Pr_{Y}[\AAA_{r,\vec{b}}(Y) \text{ succeeds}]<\sum_{r,\vec{b}}\Pr[r,\vec{b}]\cdot\frac{2^{-T}}{4}=\frac{2^{-T}}{4}$, which is a contradiction. Here $r$ denotes the randomness of $\AAA$ and $\vec{b}$ is the advice string.}

That is, after fixing $\AAA$'s coins and the advice string $\vec{b}$ as above, there is a subset of inputs $\BbB$ of size $|\BbB|\geq\frac{2^{-T}}{4}{|X| \choose T}$ such that for every $Y\in \BbB$, when executed on $Y$, our thought experiment outputs $\hat{Y}=Y$. Finally, note that the inner state  of algorithm $\AAA$ at the end of Step~1 {\em determines} the outcome of our thought experiment. Hence, as there are at least $\frac{2^{-T}}{4}{|X| \choose T}$ different outputs, there must be at least $\frac{2^{-T}}{4}{|X| \choose T}$ possible different inner states for algorithm $\AAA$, meaning that its space complexity (in bits) is at least $\log\left( \frac{2^{-T}}{4}{|X| \choose T} \right),$ which is more than $T$ provided that $|X|=\Omega(T)$ (large enough).
\end{proof}

\section{Adversarial Streaming with Bounded Interruptions (ASBI)}

In this section we present our results for the ASBI model, defined in Section~\ref{sec:ASBIintro}. We begin with our generic transformation.

\subsection{A Generic Construction for the ASBI Model}

Our construction is specified in algorithm \texttt{RobustInterruptions}. The following theorem specifies its properties.

\begin{theorem}\label{thm:interruptions}
Fix any function $g$ and fix $\alpha,\beta>0$. Let $\AAA$ be an oblivious streaming algorithm for $g$ that uses space $s$ and guarantees accuracy $\alpha$ with failure probability $\beta$. Then there exists an adversarially robust streaming algorithm for $g$ that resists $R$ interruptions and guarantees accuracy $5\alpha$ with failure probability $O(R\beta)$ using space $O(R s)$.
\end{theorem}

\begin{algorithm*}[t!]

\caption{\bf \texttt{RobustInterruptions}}
\begin{flushleft}

{\bf Input:} Parameter $R$ bounding the number of possible interruptions.

{\bf Algorithm used:} An oblivious streaming algorithm $\AAA$ with space $s$, accuracy $\alpha$, and confidence $\beta$.
\end{flushleft}

\begin{enumerate}[leftmargin=15pt,rightmargin=10pt,itemsep=1pt,topsep=3pt]

\item 
Initialize $2R$ independent instances of algorithm $\AAA$, denoted as $\AAA^{\rm answer}_1,\dots,\AAA^{\rm answer}_R$ and $\AAA^{\rm check}_1,\dots,\AAA^{\rm check}_R$. Set $r=1$.

\item For $i=1,2,\dots,m$:

\begin{enumerate}
	\item Obtain the next item in the stream $x_i\in X$.
	
	\item Feed $x_i$ to {\em all} of the copies of algorithm $\AAA$.
	
	\item Let $z_{r,i}^{\rm answer}$ and $z_{r,i}^{\rm check}$ denote the answers returned by $\AAA^{\rm answer}_r$ and $\AAA^{\rm check}_r$, respectively.
	
	\item\label{step:check} If $z_{r,i}^{\rm answer}\in(1\pm2\alpha)\cdot z_{r,i}^{\rm check}$ then output $z_{r,i}^{\rm answer}$. Otherwise, output $z_{r,i}^{\rm check}$ and set $r\leftarrow r+1$.
	
	\item If $r>R$ then {\rm FAIL}. Otherwise continue to the next iteration.

\end{enumerate}

\end{enumerate}
\end{algorithm*}

Fix an adversary $\BBB$ and consider the interaction between algorithm \texttt{RobustInterruptions} and the adversary $\BBB$. 
For $r\in[R]$, let $i_r$ denote the time step in which $z_{r,{i_r}}^{\rm check}$ is returned.

\begin{lemma}\label{lem:check}
Fix $r^*\in[R]$. 
With probability at least $1-\beta$, the answers returned by $\AAA^{\rm check}_{r^*}$ in times $1,2,\dots,i_{r^*}$ are $\alpha$-accurate. That is, for every $1\leq i\leq i_{r^*}$ it holds that $z_{r^*,i}^{\rm check}\in(1\pm\alpha)\cdot g(x_1,\dots,x_i)$. 
\end{lemma}

\begin{proof}
For simplicity, we assume that the adversary $\BBB$ is deterministic (this is without loss of generality by a simple averaging argument). Fix the randomness of all copies of algorithm $\AAA$, except for $\AAA^{\rm check}_{r^*}$. 
Let $\texttt{RI}_{r^*}$ be a variant of algorithm \texttt{RobustInterruptions} which is identical to \texttt{RobustInterruptions} until the time step $i^*$ in which $r$ becomes $r^*$. In times $i\geq i^*$, algorithm $\texttt{RI}_{r^*}$ simply outputs $z_{r^*,i}^{\rm answer}$, i.e., the answer given by $\AAA^{\rm answer}_{r^*}$. 
Note that $\mathcal{A}_{r^*}^{check}$ does not exist in algorithm \texttt{$RI_{r^*}$}.

As we fixed the coins of the copies of $\AAA\neq\AAA^{\rm check}_{r^*}$, the interaction between $\BBB$ and $\texttt{RI}_{r^*}$ is deterministic. In particular, it generates a single stream $\vec{x}_{r^*}$. 
By the utility guarantees of algorithm $\AAA^{\rm check}_{r^*}$, when run on this stream, then with probability at least $1-\beta$ it maintains $\alpha$-accuracy throughout the stream.

The lemma now follows by observing that until time $i_{r^*}$ the stream generated by the interaction between $\BBB$ and algorithm \texttt{RobustInterruptions} is identical to the stream $\vec{x}_{r^*}$.
\end{proof}

\begin{lemma}\label{lem:answer}
With probability at least $1-R\beta$, all of the answers given by \texttt{RobustInterruptions} (before returning FAIL) are $5\alpha$-accurate.
\end{lemma}

\begin{proof}
Follows from a union bound over Lemma~\ref{lem:check}, and by Step~\ref{step:check} of  \texttt{RobustInterruptions}.
\end{proof}

\begin{lemma}\label{lem:fail}
Algorithm \texttt{RobustInterruptions} returns FAIL with probability at most $2R\beta$.
\end{lemma}

\begin{proof}
Let $j_1,j_2,\dots,j_R$ denote the time steps in which the adversary conducts interruptions. That is, $j_1$ is the first time in which the adversary switches the suffix of the stream, $j_2$ is the second time this happens, and so on. Also let $p_1,p_2,\dots,p_R$ denote the time steps in which the parameter $r$ increases during the execution of algorithm \texttt{RobustInterruptions}. Specifically, $p_{\ell}$ is the time $i$ in which $r$ becomes equal to $\ell+1$. We show that for every $r\in[R]$, with probability at least $1-2r\beta$ it holds that $j_r\leq p_r$. (That is, interruptions happen ``faster'' then $r$ increases.)

The proof is by induction on $r$. For the base case, $r=1$, let $\vec{x}_1$ denote the first stream chosen by the adversary. By the utility guarantees of $\AAA$, with probability at least $1-2\beta$ we have that both $\AAA^{\rm answer}_1$ and $\AAA^{\rm check}_1$ are $\alpha$-accurate w.r.t.\ this stream, in which case $r$ does not increase. Thus, with probability at least $1-2\beta$ we have $j_1\leq p_1$.

The inductive step is similar: Fix $r\in[R]$, and suppose that $j_r\leq p_r$, which happens with probability at least $(1-2r\beta)$ by the inductive assumption.  Let $\vec{x}_r$ denote the last stream specified by the adversary before time $p_r$.
Note that the internal coins of $\AAA^{\rm answer}_{r+1}$ and $\AAA^{\rm check}_{r+1}$ are independent with this stream. Hence,
by the utility guarantees of $\AAA$, with probability at least $1-2\beta$ we have that both $\AAA^{\rm answer}_{r+1}$ and $\AAA^{\rm check}_{r+1}$ are $\alpha$-accurate w.r.t.\ this stream, in which case $r$ does not increase. Overall, with probability at least $1-2(r+1)\beta$ we have $j_{r+1}\leq p_{r+1}$.

The lemma now follows by recalling that there are at most $R$ interruptions throughout the execution. Hence, with probability at least $1-2 R\beta$ it holds that $r$ never increases beyond $R$, and the algorithm does not fail.
\end{proof}

Theorem~\ref{thm:interruptions} now follows by combining Lemmas~\ref{lem:check},~\ref{lem:answer},~\ref{lem:fail}.

\subsection{A Negative Result for the ASBI Model}

\begin{theorem}\label{thm:BInegative}
For every $R$, there exists a streaming problem over domain of size $\poly(R)$ and
stream length $\poly(R)$ that requires at least $\Omega(R)$ space to be solved in the ASBI model with $R$ interruptions to within
constant accuracy (small enough), but can be solved in the oblivious setting using space $\polylog(R)$.
\end{theorem}

This theorem follows by revisiting the negative result of Kaplan et al.~\cite{KaplanMNS21} for the (plain) adversarial model. They presented a streaming problem, called the SADA problem, that is easy to solve in the oblivious setting but requires large space to be solved in the adversarial setting. To obtain their hardness results, \cite{KaplanMNS21} showed a reduction from a hard problem in learning theory (called the {\em adaptive data analysis (ADA) problem}) to the task of solving the SADA problem in the adversarial setting with small space. 

In the ADA problem, the goal is to design a mechanism $\AAA$ that initially obtains a dataset $D$ containing $n$ i.i.d.\ samples from some unknown distribution $\PPP$, and then answers $k$ {\em adaptively chosen queries} w.r.t.\ $\PPP$. Importantly, $\AAA$'s answers must be accurate w.r.t.\ the underlying distribution $\PPP$, and not just w.r.t.\ the empirical dataset $D$.
Hardt, Ullman, and Steinke~\cite{HardtU14,DBLP:conf/ita/SteinkeU16} showed that the ADA problem requires a large sample complexity. Specifically, they showed that every efficient\footnote{The results of \cite{HardtU14,DBLP:conf/ita/SteinkeU16} hold for all computationally efficient mechanisms, or alternatively, for a class of unbounded mechanisms which they call {\em natural} mechanisms.} mechanism for this problem must have sample complexity $n\geq\Omega(\sqrt{k})$.

Theorem~\ref{thm:BInegative} follows by the following two observations regarding the negative result of \cite{KaplanMNS21} for the SADA problem, and regarding the underlying hardness result of \cite{HardtU14,DBLP:conf/ita/SteinkeU16} for the ADA problem:
\begin{enumerate}
    \item In the hardness results of \cite{HardtU14,DBLP:conf/ita/SteinkeU16} for the ADA problem, the adversary generates the queries using $O(n)$ rounds of adaptivity, where $n$ is the sample size. In more detail, even though the adversary poses $\poly(n)$ queries throughout the interaction\footnote{The adversary poses $O(n^3)$ in \cite{HardtU14} and $O(n^2)$ in \cite{DBLP:conf/ita/SteinkeU16}.}, these queries are generated in $O(n)$ bulks where queries in the $j$th bulk depend only on answers given to queries of {\em previous} bulks.
    \item The reduction of Kaplan et al.~\cite{KaplanMNS21} from the ADA problem to the SADA problem maintains the number of adaptivity rounds. That is, the reduction of Kaplan et al.~\cite{KaplanMNS21} transforms an adversary for the ADA problem that generates the queries in $\ell$ bulks into an adversary for the SADA problem that uses $\ell$ interruptions.
\end{enumerate}

We remark that Theorem \ref{thm:BInegative} holds even for a model in which the streaming algorithm is strengthen and gets an indication during each interruption round.
That is true since by the technique of \cite{KaplanMNS21}, the streaming algorithm can identify the exact round of a new bulk and  such round corresponds to an interruption round.

In Appendix~\ref{sec:surveyADASADA} we survey the necessary details from \cite{KaplanMNS21,HardtU14,DBLP:conf/ita/SteinkeU16}, and provide a more detailed account of the modifications required in order to obtain Theorem~\ref{thm:BInegative}.

\bibliographystyle{abbrv}
\bibliography{bib}

\begin{thebibliography}{10}

\bibitem{AhnGM12}
K.~J. Ahn, S.~Guha, and A.~McGregor.
\newblock Analyzing graph structure via linear measurements.
\newblock In {\em SODA}, pages 459--467, 2012.

\bibitem{AhnGM12b}
K.~J. Ahn, S.~Guha, and A.~McGregor.
\newblock Graph sketches: sparsification, spanners, and subgraphs.
\newblock In {\em PODS}, pages 5--14, 2012.

\bibitem{AlonMS99}
N.~Alon, Y.~Matias, and M.~Szegedy.
\newblock The space complexity of approximating the frequency moments.
\newblock {\em J. Comput. Syst. Sci.}, 58(1):137--147, 1999.

\bibitem{ACSS21}
I.~Attias, E.~Cohen, M.~Shechner, and U.~Stemmer.
\newblock A framework for adversarial streaming via differential privacy and
  difference estimators.
\newblock {\em CoRR}, abs/2107.14527, 2021.

\bibitem{BassilyNSSSU16}
R.~Bassily, K.~Nissim, A.~D. Smith, T.~Steinke, U.~Stemmer, and J.~Ullman.
\newblock Algorithmic stability for adaptive data analysis.
\newblock In {\em STOC}, pages 1046--1059, 2016.

\bibitem{BNS13b}
A.~Beimel, K.~Nissim, and U.~Stemmer.
\newblock Private learning and sanitization: Pure vs. approximate differential
  privacy.
\newblock In {\em APPROX-RANDOM}, pages 363--378, 2013.

\bibitem{BEO22}
O.~Ben{-}Eliezer, T.~Eden, and K.~Onak.
\newblock Adversarially robust streaming via dense-sparse trade-offs.
\newblock In {\em SOSA@SODA}, pages 214--227, 2022.

\bibitem{BenEliezerJWY20}
O.~Ben{-}Eliezer, R.~Jayaram, D.~P. Woodruff, and E.~Yogev.
\newblock A framework for adversarially robust streaming algorithms.
\newblock {\em J. {ACM}}, 69(2):17:1--17:33, 2022.

\bibitem{BenEliezerY19}
O.~Ben{-}Eliezer and E.~Yogev.
\newblock The adversarial robustness of sampling.
\newblock In {\em PODS}, pages 49--62, 2020.

\bibitem{BunDRS18}
M.~Bun, C.~Dwork, G.~N. Rothblum, and T.~Steinke.
\newblock Composable and versatile privacy via truncated {CDP}.
\newblock In {\em STOC}, pages 74--86, 2018.

\bibitem{BNSV15}
M.~Bun, K.~Nissim, U.~Stemmer, and S.~P. Vadhan.
\newblock Differentially private release and learning of threshold functions.
\newblock In {\em {FOCS}}, pages 634--649, 2015.

\bibitem{DBLP:conf/icml/Cohen0NSSS22}
E.~Cohen, X.~Lyu, J.~Nelson, T.~Sarl{\'{o}}s, M.~Shechner, and U.~Stemmer.
\newblock On the robustness of countsketch to adaptive inputs.
\newblock In K.~Chaudhuri, S.~Jegelka, L.~Song, C.~Szepesv{\'{a}}ri, G.~Niu,
  and S.~Sabato, editors, {\em International Conference on Machine Learning,
  {ICML} 2022, 17-23 July 2022, Baltimore, Maryland, {USA}}, volume 162 of {\em
  Proceedings of Machine Learning Research}, pages 4112--4140. {PMLR}, 2022.

\bibitem{CohenLNSS23}
E.~Cohen, X.~Lyu, J.~Nelson, T.~Sarl{\'{o}}s, and U.~Stemmer.
\newblock {\~{O}}ptimal differentially private learning of thresholds and
  quasi-concave optimization.
\newblock {\em CoRR}, abs/2211.06387, 2022.

\bibitem{DFHPRR14}
C.~Dwork, V.~Feldman, M.~Hardt, T.~Pitassi, O.~Reingold, and A.~L. Roth.
\newblock Preserving statistical validity in adaptive data analysis.
\newblock In {\em STOC}, pages 117--126, 2015.

\bibitem{DMNS06}
C.~Dwork, F.~McSherry, K.~Nissim, and A.~Smith.
\newblock Calibrating noise to sensitivity in private data analysis.
\newblock In {\em TCC}, pages 265--284, 2006.

\bibitem{DRV10}
C.~Dwork, G.~N. Rothblum, and S.~P. Vadhan.
\newblock Boosting and differential privacy.
\newblock In {\em FOCS}, pages 51--60, 2010.

\bibitem{GHRSW12}
A.~C. {Gilbert}, B.~{Hemenway}, A.~{Rudra}, M.~J. {Strauss}, and M.~{Wootters}.
\newblock Recovering simple signals.
\newblock In {\em Information Theory and Applications Workshop (ITA)}, pages
  382--391, 2012.

\bibitem{GHSWW12}
A.~C. {Gilbert}, B.~{Hemenway}, M.~J. {Strauss}, D.~P. {Woodruff}, and
  M.~{Wootters}.
\newblock Reusable low-error compressive sampling schemes through privacy.
\newblock In {\em IEEE Statistical Signal Processing Workshop (SSP)}, pages
  536--539, 2012.

\bibitem{HardtU14}
M.~Hardt and J.~Ullman.
\newblock Preventing false discovery in interactive data analysis is hard.
\newblock In {\em FOCS}. IEEE, October 19-21 2014.

\bibitem{HardtW13}
M.~Hardt and D.~P. Woodruff.
\newblock How robust are linear sketches to adaptive inputs?
\newblock In {\em STOC}, pages 121--130, 2013.

\bibitem{HKMMS-JACM}
A.~Hassidim, H.~Kaplan, Y.~Mansour, Y.~Matias, and U.~Stemmer.
\newblock Adversarially robust streaming algorithms via differential privacy.
\newblock {\em J. ACM}, 2022.

\bibitem{JowhariST11}
H.~Jowhari, M.~Saglam, and G.~Tardos.
\newblock Tight bounds for lp samplers, finding duplicates in streams, and
  related problems.
\newblock In {\em {PODS}}, pages 49--58. {ACM}, 2011.

\bibitem{jowhari2011tight}
H.~Jowhari, M.~Sa{\u{g}}lam, and G.~Tardos.
\newblock Tight bounds for lp samplers, finding duplicates in streams, and
  related problems.
\newblock In {\em Proceedings of the thirtieth ACM SIGMOD-SIGACT-SIGART
  symposium on Principles of database systems}, pages 49--58, 2011.

\bibitem{KaplanLMNS19}
H.~Kaplan, K.~Ligett, Y.~Mansour, M.~Naor, and U.~Stemmer.
\newblock Privately learning thresholds: Closing the exponential gap.
\newblock In {\em COLT}, pages 2263--2285, 2020.

\bibitem{KaplanMNS21}
H.~Kaplan, Y.~Mansour, K.~Nissim, and U.~Stemmer.
\newblock Separating adaptive streaming from oblivious streaming using the
  bounded storage model.
\newblock In {\em {CRYPTO}}, pages 94--121, 2021.

\bibitem{MT07}
F.~McSherry and K.~Talwar.
\newblock Mechanism design via differential privacy.
\newblock In {\em FOCS}, pages 94--103, 2007.

\bibitem{MironovNS11}
I.~Mironov, M.~Naor, and G.~Segev.
\newblock Sketching in adversarial environments.
\newblock {\em {SIAM} J. Comput.}, 40(6):1845--1870, 2011.

\bibitem{DBLP:conf/ita/SteinkeU16}
T.~Steinke and J.~R. Ullman.
\newblock Interactive fingerprinting codes and the hardness of preventing false
  discovery.
\newblock In {\em 2016 Information Theory and Applications Workshop, {ITA}
  2016, La Jolla, CA, USA, January 31 - February 5, 2016}, pages 1--41. {IEEE},
  2016.

\bibitem{StemmerPhD}
U.~Stemmer.
\newblock {\em Individuals and privacy in the eye of data analysis}.
\newblock PhD thesis, Ben-Gurion University of the Negev, 2016.
\newblock Supervisors -- Amos Beimel and Kobbi Nissim.

\bibitem{WZ21}
D.~P. Woodruff and S.~Zhou.
\newblock Tight bounds for adversarially robust streams and sliding windows via
  difference estimators.
\newblock In {\em FOCS}, pages 1183--1196, 2022.

\end{thebibliography}

\appendix
\section{Details for Theorem~\ref{thm:BInegative} }\label{sec:surveyADASADA}
In this section we elaborate on the components from which Theorem~\ref{thm:BInegative} follows:
\begin{enumerate}
    \item The hardness results of \cite{HardtU14,DBLP:conf/ita/SteinkeU16} for the ADA problem.
    \item The reduction of Kaplan et al.~\cite{KaplanMNS21} from the ADA problem to the SADA problem.
\end{enumerate}

The purpose is to show that the number of {\em adaptive} rounds of the reduction is $O(R)$ where $R$ is the bound on the interruptions, in spite of the fact that the number of queries throughout the attack of \cite{DBLP:conf/ita/SteinkeU16} is $O(R^2)$ and despite the fact that, in the reduction of \cite{KaplanMNS21}, these $O(R^2)$ queries are encoded and delivered to the streaming algorithm using a somewhat long stream of length $\poly(R)$. 
In other words, we show that the negative result presented by \cite{KaplanMNS21} for algorithms that solve the SADA problem is in fact stronger in the sense that it rules out algorithms in the ASBI model, and not only algorithms in the plain adversarial model.

\subsection{IFPC adaptivity level} We now elaborate on the number of { \em adaptive} rounds in the hardness results \cite{HardtU14,DBLP:conf/ita/SteinkeU16} (as apposed to the total number of rounds). Specifically, \cite{DBLP:conf/ita/SteinkeU16} presents a two player game protocol (see definition below) between the players {\em adversary} $\mathcal{P}$ and the finger printing code $\mathcal{F}$. In their paper, \cite{DBLP:conf/ita/SteinkeU16} use this definition of a game along with a code, namely {\em Interactive Finger Printing Code} (denote IFPA) to prove an upper bound on the number of accurate queries that can be guaranteed against an adaptive analyst. The role of the analyst is played by $\mathcal{F}$. \cite{DBLP:conf/ita/SteinkeU16} present an algorithm for $\mathcal{F}$ that assures that $\mathcal{P}$ looses after $O(n^2)$ adaptive rounds (i.e. return an inaccurate answer to a query of the analyst), where $n$ is the size of the database in the game. The game is defined as follows:\\

\noindent\fbox{%
    \parbox{\linewidth}{%
    {\bf Game protocol: adversary $\PPP$ vs IFPC $\FFF$}
        \begin{enumerate}
            \item $\PPP$ selects a subset $S^1 \subseteq [N]$, unknown to $\FFF$. 
            \item For $j = 1,\dots, \ell$: 
            \begin{enumerate}
                \item $\FFF$ outputs a column vector $c^j\in \{\pm1\}^N$
                \item Let $c^j_{S^j}\in \{\pm1\}^{|S^j|}$ be a restriction of $C^j$ to coordinates $S^j$, which is given to $\PPP$.\label{gamestep:restrictCodeWord}
                \item $\PPP$ outputs $a^j\in \{\pm1\}$, which is given to $\FFF$.
                \item $\FFF$ accuses (possibly empty) set of users $I^j\subseteq [N]$. Let $S^{j+1}=S^{j}\setminus I^{j}$.
            \end{enumerate}
        \end{enumerate}
    }%
}
\;

In above game, $\PPP$ is defined as a coalition of the users $S^1$, that receives in each round only a partial code word $C^j$. The goal of $\PPP$ is to remain {\em consistent}. That means that whenever the query $c_j$ is all $+1$ or all $-1$, then $\PPP$ must answer $+1$ or $-1$ correspondingly. On each round, $\FFF$ chooses some subset $I^j\subseteq [N]$ to accuse. That means that in the next round these entries will be also restricted from $\PPP$. The goal of $\FFF$ is to make $\PPP$ be inconsistent while it cannot accuse "too many" users that are not from $\PPP$ (i.e.~users $S^1$). 

In their paper \cite{DBLP:conf/ita/SteinkeU16} show an algorithm for $\FFF$ that assures inconsistency of any $\PPP$ after $O(n^2)$ number of rounds (see algorithm \ref{alg:IFPC}).
\begin{algorithm}[ht]
\caption{\bf \texttt{IFPC}($n, N, \delta, \beta$) \cite{DBLP:conf/ita/SteinkeU16}} \label{alg:IFPC}
\begin{flushleft}

{\bf Input:} Parameters: $N$ is the number of users, $1\leq n \leq N$ is the size of the coalition, $\delta \in (0,1]$ is the failure probability, $\beta < 1/2$ is the fraction of allowed inconsistent rounds.

\end{flushleft}

\begin{enumerate}[leftmargin=15pt,rightmargin=10pt,itemsep=1pt,topsep=3pt]

\item Set parameters $\alpha=(1/2-\beta)/4n=\Omega(1/n)$, $\zeta=3/8-\beta/4 = 1/2-1/4(1/2-\beta)$, $\sigma = O((n/(1/2-\beta)^2)/\log(\delta^{-1}))$, $\ell = O((n^2/(1/2-\beta)^4)\log(1/\delta))$

\item Let $s_i^0 = 0$ for every $i\in [N]$

\item for $j = 1,\dots , \ell$:
    \begin{enumerate}
    
        \item Draw $p_j\sim \overline{D_{\alpha,\zeta}}$ and $c^j_{1\dots N}\sim p^j$.
        
        \item Issue $c^j\in \{\pm 1\}^N$ as a challenge and receive $a^j\in \{\pm 1\}$ as a response.
        
        \item For $i\in [N]$, let $s_i^j = s_i^{j-1} + a_j\cdot \phi^{p^j}(c_i^j)$.
        
        \item Accuse $I^j = \left\{ i\in [N] | s_i^j >\sigma \right\}$
    \end{enumerate}

\end{enumerate}
\end{algorithm}
In that algorithm (\ref{alg:IFPC}), $\overline{D_{\alpha,\zeta}}$ is a distribution over $[0,1]$ from which a Bernoulli parameter $p^j$ is drawn and used to generate the $j$th code word $c_{1\dots N}^j\in p^j$ and $\phi^{p^j}:\{0,1\}\rightarrow \R$ is a function that measures a correlation quantity between user $i$'s input $c^j_i$ and the output $a^j$. The main idea is to accumulate for each user this correlation quantity over the iterations, and once this quantity crossing some threshold for a user, then the algorithm decides that the user is a part of the coalition and is marked as such (accused).

Importantly to our use case, note that in algorithm \ref{alg:IFPC} all of the code words $c^j$ can be drawn in advance. That is, they are independent from the answers $a^j$. The code words are not presented to $\PPP$ all at once but one by one, and so the algorithm \ref{alg:IFPC} is {\em interactive}, while it is {\em not adaptive}.

Yet the game protocol itself, {\em is} adaptive. In the protocol step~\ref{gamestep:restrictCodeWord}, the part of the code word that is sent to $\PPP$ is restricted only to the coalition users $S^j$. That is, the users from the initial coalition $S^1$ that have not yet been accused, which is a function of all previous answers $a^j$.
And so, the adaptivity is reflected by the times that code words that are given to $\PPP$ are determined. 

We now conclude that the number of such determining times is only $O(n)$: The list of coalition users that are not accused $S^j$ is  monotonic decreasing, thus can be updated at most $n$ times. 
Now, recall that the algorithm has length of $O(n^2)$ iterations. 
In addition the set of code words that are given to $\PPP$ between two consecutive modifications of $S^j$ is fixed.
And so, denote by $j_1,\dots,j_k$ as the times that the list $S^j$ is modified (for some $k\leq n$), then for $i\in [k]$ during the iterations $[j_i,\dots,j_{i+1})$ the game protocol is {\em not} adaptive. 

\subsection{ADA to SADA reduction maintains adaptivity level} Now we look on the hardness result of Kaplan et al.~\cite{KaplanMNS21}. In their paper, Kaplan et al.~ show a reduction from the ADA problem, that is shown to be hard (a bound of $O(n^2)$ query rounds) in \cite{DBLP:conf/ita/SteinkeU16}, to SADA problem (Streaming Adaptive Data Analysis). In high level, the idea has the following components: \begin{enumerate}
    \item {\bf Stream generation and a streaming algorithm:} The stream is determined w.r.t.~the game protocol of IFPC where $\PPP$ is a streaming algorithm that answers queries (encoded in the input stream).
    \item {\bf Compression:} Algorithms with small space (significantly smaller than the size of their input dataset) are known to have strong generalization properties. Hence, if a small space algorithm is solving the SADA problem (for number of queries $\gg O(n^2)$ and with space $\ll n$) then it must, in fact, solve the underlying statistical ADA problem for the same number of queries, contradicting \cite{DBLP:conf/ita/SteinkeU16}.
\end{enumerate}

These two components show that any algorithm for the SADA problem must have space $\Omega(n)$. 
Since the number of {\em bounded interruptions} in our lower bound paradigm is the number of {\em adaptive}  rounds that is determined in the first component, we elaborate on that component only.

The stream is defined w.r.t.~the IFPC game protocol in two stages. 
First stage is setting the set $S^1$ via $n$ updates in the stream. Each stream update encode a single user in $S^1$. 
Then on the second stage, for $O(n^2)$ game protocol rounds, the restricted query $c^j_{S^j}$ is sent for the streaming algorithm for an answer. A small issue is that each of these queries is of encoded size of $\poly(n)$, and thus it takes $\poly(n)$ stream updates for the streaming algorithm to receive it (and so, the stream length of the attack is of $\poly(n)$). 
After each such (encoded) query, the streaming algorithm must answer (by \cite{DBLP:conf/ita/SteinkeU16}, any such algorithm must fail after $O(n^2)$ queries).

Yet, the observation that we have only $O(n)$ {\em adaptive} rounds remains, since the reduction uses the game protocol of \cite{DBLP:conf/ita/SteinkeU16}.
And so, a similar reduction holds from the ADA problem to the SADA problem in the ASBI model with $O(n)$ rounds of interruptions. This implies a lower bound on the space of this model.
\end{document}